\documentclass[11pt,a4paper]{article}

\usepackage[cp1251]{inputenc}
\usepackage[T2A]{fontenc}

\usepackage[pdftex]{graphicx}
\usepackage{amsmath}
\usepackage{amssymb}
\usepackage{amsxtra}
\usepackage{amsthm}
\usepackage{latexsym}
\usepackage{array}
\usepackage{multirow}
\usepackage{hhline}
\usepackage[in]{fullpage}
\usepackage{indentfirst}

\def\sgrad{\mathrm{sgrad}\,}
\def\trace{\mathrm{trace}\,}
\usepackage{hyperref}
\theoremstyle{plain}
\newtheorem{theorem}{Theorem}
\begin{document}

\begin{center}
\begin{Large}
\textbf{BIFURCATION DIAGRAM OF ONE GENERALIZED INTEGRABLE MODEL OF VORTEX DYNAMICS}
\vskip 5mm
\textbf{Pavel\,\,E.~Ryabov$^{1,2,3}$ and Artemiy\,\,A.~Shadrin$^{1}$}
\end{Large}

${}^{1}$\! Financial University under the Government of the Russian Federation\\
Leningradsky prosp. 49, Moscow, 125993 Russia

${}^{2}$\! Institute of Machines Science, Russian Academy of Sciences\\
Maly Kharitonyevsky per. 4, Moscow, 101990 Russia

${}^{3}$\! Udmurt State University\\
ul. Universitetskaya 1, Izhevsk, 426034 Russia%

E-mail: PERyabov@fa.ru, shadrin.art@gmail.com

\end{center}

\begin{abstract}
The article is devoted to the results of a phase topology research on a generalized mathematical model, which covers such two problems as dynamics of two point vortices enclosed in a harmonic trap in a Bose-Einstein condensate and dynamics of two point vortices bounded by a circular region in an ideal fluid. New bifurcation diagrams are obtained and three-into-one and four-into-one tori bifurcations are observed for some values of the model's physical parameters. The presence of such bifurcations in the integrable model of vortex dynamics with positive intensities indicates a complex transition and a connection between bifurcation diagrams in both limiting cases. In this paper, we analytically derive the equations, that define the parametric family of the generalized model's bifurcation diagrams, including bifurcation diagrams of the specified limiting cases. The dynamics of the general case's bifurcation diagram is shown, using its implicit parametrization. The stable bifurcation diagram, related to the problem of dynamics of two vortices bounded by a circular region in an ideal fluid, is observed for particular values of parameters.
\end{abstract}

\noindent\textit{Keywords}:\,{completely integrable Hamiltonian system, bifurcation diagram, bifurcation of\\ Liouville tori, dynamics of point vortices, Bose-Einstein condensate}
\textit{Received April 19, 2019}

\def\sgrad{\mathrm{sgrad}\,}
\def\trace{\mathrm{trace}\,}

\section{INTRODUCTION}

Integrable models of point vortices on a plane hold a central position in the analytical dynamics of vortex structures. Studies of vortex dynamics in quantum physics have shown that quantum vortices behave similarly to the thin vortex filaments studied in classical fluid dynamics. A special attention is paid to vortex structures in a Bose-Einstein condensate obtained for ultracold atomic gases \cite{fett2009}. In this paper we consider
a generalized mathematical model which covers such two problems as dynamics of two point vortices enclosed in a harmonic trap in a Bose-Einstein condensate \cite{kevrekPhysLett2011}, \cite{kevrek2013}, \cite{kevrikidis2014} and dynamics of two point vortices bounded by a circular region in an ideal fluid \cite{Greenhill1877}, \cite{BorMamSokolovskii2003}, \cite{BorMam2005book}.
This model leads to a completely Liouville integrable Hamiltonian system with two degrees of freedom, therefore topological methods used for such systems can be applied. Topological methods have been successfully used to analyze the stability of absolute and relative choreographies \cite{BorMamSokolovskii2003}, \cite{borkil2000},  \cite{bormamkil2004}, \cite{kilinbormam2013}, \cite{BorSokRyab2016}. In integrable models, these motions are usually associated with first integrals' constant values for which these integrals, that are considered as functions of phase variables, are dependent in the sense of a linear dependence of differentials. The main role in the study of such dependence is played by a bifurcation diagram of momentum mapping.

The generalized mathematical model is described by a Hamiltonian system of differential equations:
\begin{equation}
\label{eq1_1}
\displaystyle{\Gamma_k\dot x_k=\frac{\partial H}{\partial y_k} (z_1,z_2);\quad \Gamma_k\dot y_k=-\frac{\partial H}{\partial x_k} (z_1,z_2),\quad k=1,2,}
\end{equation}
where the Hamiltonian $H$ has the form
\begin{equation}
\label{eq1_2}
\begin{array}{l}
\displaystyle{H=\frac{1}{2}\Bigl[\Gamma_1^2\ln(1-|z_1|^2)+\Gamma_2^2\ln(1-|z_2|^2)+\Gamma_1\Gamma_2\ln\left(\frac{[|z_1-z_2|^2+(1-|z_1|^2)(1-|z_2|^2)]^\varepsilon}{|z_1-z_2|^{2(c+\varepsilon)}}\right)\Bigr].}
\end{array}
\end{equation}
Here, the Cartesian coordinates of $k$-th vortex ($k=1,2$) with intensities $\Gamma_k$ are denoted by\linebreak
$z_k=x_k+{\rm i}y_k$. Physical parameter ``$c$'' expresses the extent of the vortices' interaction,
$\varepsilon$ is a parameter of deformation. These parameters determine two limiting cases, namely, the model of two point vortices enclosed in a harmonic trap in a Bose-Einstein condensate ($\varepsilon=0$) \cite{kevrekPhysLett2011}, \cite{kevrek2013}, \cite{kevrikidis2014} and the model of two point vortices bounded by a circular region in an ideal fluid ($c=0$, $\varepsilon=1$)  \cite{Greenhill1877}, \cite{BorMamSokolovskii2003}, \cite{BorMam2005book}.
The phase space $\cal P$ is defined as a direct product of two open disks of radius $1$ with the exception of vortices' collision points
\begin{equation*}
{\cal P}=\{(z_1,z_2)\,:\, |z_1|<1,\,|z_2|<1,z_1\ne z_2\}.
\end{equation*}
 The Poisson structure on the phase space $\cal P$ is given in the standard form
\begin{equation}
\label{eq1_3}
\{z_k,\bar{z}_j\}=-\frac{2\rm i}{\Gamma_k}\delta_{kj},
\end{equation}
 where $\delta_{kj}$ is the Kronecker delta.

The system $\eqref{eq1_1}$ admits an additional first integral of motion, \textit{the angular momentum of vorticity},
\begin{equation}
\label{eq1_4}
F=\Gamma_1|z_1|^2+\Gamma_2|z_2|^2.
\end{equation}

The function $F$ together with the Hamiltonian $H$ forms on $\cal P$ a complete involutive set of integrals of system $\eqref{eq1_1}$. According to the Liouville-Arnold theorem, a regular level surface of the first integrals is a nonconnected union of two-dimensional tori filled with conditionally periodic trajectories. The \textit{momentum mapping} ${\cal F}\,:\, {\cal P}\to {\mathbb R}^2$ is defined by setting  ${\cal F}(\boldsymbol x)=(F(\boldsymbol x), H(\boldsymbol x))$. Let $\cal C$ denote the set of all critical points of the momentum mapping, i.e., points at which $\mathop{\rm rank}\nolimits d{\cal F}(\boldsymbol x) < 2$. The set of critical values $\Sigma = {\cal F}({\cal C}\cap{\cal P})$ is called the \textit{bifurcation diagram}.

In works \cite{SokRyabRCD2017} and \cite{sokryab2018} the bifurcation diagram was analytically investigated at $c=1$ and $\varepsilon=0$. In \cite{RyabDan2019} and \cite{RyabSocND2019} a reduction to a system with one degree of freedom was performed and a bifurcation of three tori into one was found at $c>3$ and $\varepsilon=0$.
This bifurcation was observed earlier by Kharlamov \cite{Kharlamov1988} while studying a phase topology of the Goryachev-Chaplygin-Sretensky integrable case in rigid body dynamics. In Fomenko, Bolsinov, and Matveev's work \cite{bolsmatvfom1990} it was found as a singularity in a 2-atom form of a Liouville foliation's singular layer. In Oshemkov and Tuzhilin's work \cite{oshtuzh2018}, devoted to the splitting of saddle singularities, such a bifurcation was found to be unstable and its perturbed foliations were presented. In the situation where the physical parameter of vortices' intensity ratio is experiencing integrable perturbation, said bifurcation comes down to the bifurcation of two tori into one and vice versa \cite{RyabDan2019}. In another limiting case ($c=0, \varepsilon=1$), the bifurcation analysis of dynamics of two point vortices bounded by a circular domain in an ideal fluid is performed \cite{BorMamSokolovskii2003}, \cite{BorMam2005book}. In these limiting cases completely different bifurcation diagrams were obtained. In the case of a positive vortex pair a new bifurcation diagram is obtained for which the bifurcation of four tori into one is observed \cite{RyabovArXiv2019}. The presence of three-into-one and four-into-one tori bifurcations in the integrable model of vortex dynamics with positive intensities indicates a complex transition and connection between two bifurcation diagrams in both limiting cases.
A.\,V.~Borisov suggested to study both these integrable models and to find out how the bifurcation diagrams of both limiting cases are related. In this paper we analytically derive the equations that define a parametric family of bifurcation diagrams of the generalized model \eqref{eq1_1} containing bifurcation diagrams of the specified limiting cases. In the general case reduction to a system with one degree of freedom allows us to apply level curves of corresponding Hamiltonian in order to observe different kinds of Liouville tori bifurcations.

\section{CRITICAL SET}
\subsection{General case}
We define the polynomial expressions $F_1$ and $F_2$ from phase variables:
\begin{eqnarray}
&\label{eq2_1}
F_1=x_1y_2-y_1x_2,\\[3mm]
&
\begin{array}{l}
\label{eq2_2}
F_2= cx_2(\Gamma_1x_1+\Gamma_2x_2)(x_2^2+y_2^2-1)[x_2(x_1x_2-1)+x_1y_2^2][(x_1^2-1)x_2^2+x_1^2y_2^2]+\\[3mm]
+\Gamma_2[(x_1^2-1)x_2^2+x_1^2y_2^2]\Bigl\{\varepsilon x_2^3(x_2^2+y_2^2-1)^2+x_1(x_1-x_2)(x_2^2+y_2^2)[x_2(x_1x_2-1)+x_1y_2^2]\Bigr\}\\[3mm]
+\Gamma_1x_1(x_2^2+y_2^2-1)\Bigl\{x_2^2(x_2-x_1)(x_2^2+y_2^2)(x_2(x_1x_2-1)+x_1y_2^2)+\varepsilon[(x_1^2-1)x_2^2+x_1^2y_2^2]^2\Bigr\},
\end{array}
\end{eqnarray}
and denote by ${\cal N}$ the closure of system's set of solutions:
\begin{equation}
\label{eq2_3}
F_1 = 0,\quad F_2 = 0.
\end{equation}

Then the theorem below is true.
\begin{theorem}
The set of critical points $\cal C$ of the moment mapping $\cal F$ coincides with the \eqref{eq2_3} system's set of solutions. The set ${\cal N}$ is a two-dimensional invariant submanifold of the system \eqref{eq1_1} with the Hamiltonian \eqref{eq1_2}.
\end{theorem}

\begin{proof}
To prove the first statement of the theorem it is necessary to find the
phase space points where the rank of the moment map is not maximal. With the help of direct computations one can verify that the Jacobi matrix of the moment map has
zero minors of the second order at the points ${\boldsymbol z}\in{\cal P}$, the coordinates of which satisfy the equations of system \eqref{eq2_3}.
Therefore ${\cal C} = {\cal N}$. The fact that the relations \eqref{eq2_3} are invariate might be prooved by the following chain of correct equalities:
\begin{equation*}
\dot F_1 = \{F_1,H\}_{F_1=0} = \sigma_1F_2,\quad \dot F_2 = \{F_2,H\}_{F_1=0} = -\frac{x_2y_2}{x_2^2+y_2^2}\sigma_1\sigma_2F_2,
\end{equation*}
where polynomial functions $\sigma_k$ from phase variables have the following explicit form:
\begin{equation*}
\begin{array}{l}
\displaystyle{\sigma_1=\frac{1}{(x_1 - x_2)x_2(x_2^2 + y_2^2-1)[x_2(x_1x_2-1) + x_1y_2^2][(x_1^2-1)x_2^2 + x_1^2y_2^2]},}\\[3mm]
\sigma_2=c(x_2^2+y_2^2-1)\Bigl\{\Gamma_1\Bigl[x_2^3\Bigl(4+x_1\bigl((x_1^2-3)x_2-2x_1\bigr)\Bigr)+x_1x_2\bigl((2x_1^2-3)x_2-2x_1\bigr)y_2^2+x_1^3y_2^4\Bigr]-\\[3mm]
-\Gamma_2x_2(x_2^2+y_2^2)(-2x_1x_2-x_2^2+3x_1^2(x_2^2+y_2^2))\Bigr\}-\Gamma_2(x_2^2+y_2^2)\Bigl\{4x_2^3+x_1^3(x_2^2+y_2^2)(1+x_2^2+y_2^2)+\\[3mm]
+x_1x_2^2(2\varepsilon(x_2^2+y_2^2-1)^2-3(1+x_2^2+y_2^2))\Bigr\}-\Gamma_1x_2(x_2^2+y_2^2-1)\Bigl\{4\varepsilon[(x_1^2-1)x_2^2+x_1^2y_2^2]+\\[3mm]
+(x_2^2+y_2^2)[2x_1x_2(1+x_2^2+y_2^2)-x_2^2-3x_1^2(x_2^2+y_2^2)]\Bigr\}.
\end{array}
\end{equation*}
\end{proof}

\subsection{Some special cases}

Let
\begin{equation}
\label{eq2_40}
\Gamma_1=1, \Gamma_2=a, \varepsilon=0, c\ne 0.
\end{equation}
Substitution of \eqref{eq2_40} parameters' values in \eqref{eq2_2} leads to the expression
\begin{equation}
\label{eq3_4}
\begin{array}{l}
F_2=a(x_2^2+y_2^2)^2x_1^4-\\[3mm]
-(x_2^2+y_2^2)[(x_2^2+y_2^2)(a-c)+c]x_2x_1^3+(x_2^2+y_2^2)[(ac-1)(x_2^2+y_2^2-1)-a]x_1^2x_2^2+\\[3mm]
[(x_2^2+y_2^2)(x_2^2+y_2^2+a-c-1)+b]x_2^3x_1-ac(x_2^2+y_2^2-1)x_2^4.
\end{array}
\end{equation}

In the case of a positive vortex pair
\begin{equation}
\label{eq2_60}
\Gamma_1=1, \Gamma_2=1,\varepsilon\ne 0, c\ne 0,
\end{equation}
the equation \eqref{eq2_2} takes the form
\begin{equation*}
F_2=(x_1+x_2)F_3,
\end{equation*}
where
\begin{equation*}
\begin{array}{l}
F_3=[x_1(x_2^2+y_2^2)-x_2]\Bigl\{(x_1^2+x_2^2)(x_2^2+y_2^2)[x_1(x_2^2+y_2^2)-cx_2]+x_2[(c-2)(x_2^2+y_2^2)^2x_1^2+cx_2^2]\Bigr\}\\[3mm]
+\varepsilon(x_2^2+y_2^2-1)[x_1^2(x_2^2+y_2^2)-x_2^2][(x_2^2+y_2^2)(x_1^2-x_1x_2+x_2^2)-x_2^2].
\end{array}
\end{equation*}
The equation $F_2=0$ breaks down into two subsystems, taking into account \eqref{eq2_1}:
\begin{equation*}
x_1+x_2=0,\quad y_1+y_2=0
\end{equation*}
and
\begin{equation*}
F_1 = 0,\quad F_3 = 0,
\end{equation*}
each of which is a two-dimensional invariant submanifold of the system \eqref{eq1_1} with the Hamiltonian \eqref{eq1_2} and specified as in \eqref{eq2_60} parameters' values.

It should be noted that the expression \eqref{eq2_2} for another limiting case
\begin{equation*}
c=0,\quad \varepsilon=1,
\end{equation*}
also disintegrates in the case of a positive vortex pair ($\Gamma_1=\Gamma_2=1$):
\begin{equation*}
F_2=(x_1+x_2)F_3,
\end{equation*}
where
\begin{equation*}
\begin{array}{l}
F_3=-x_2^4(x_2^2+y_2^2-1)^2+\\[3mm]
+x_1(x_2^2+y_2^2)\Bigl\{-x_2^3-x_1^3(x_2^2+y_2^2)-x_1^2(x_2^2+y_2^2)^2(3x_2-2x_1)+x_1x_2^2[2-y_2^2-x_2^2+2(x_2^2+y_2^2)^2]\Bigr\}.
\end{array}
\end{equation*}

\section{BIFURCATION DIAGRAM}
\subsection{General case}
To determine the bifurcation diagram $\Sigma $ as the image of the critical points' set $\cal C $
of the momentum mapping $\cal F $, it is convenient to change to polar coordinates
\begin{equation}
\label{eq3_1}
x_1 = r_1\cos\theta_1,\quad y_1 = r_1\sin\theta_1,\quad
x_2 = r_2\cos\theta_2,\quad y_2 = r_2\sin\theta_2.
\end{equation}
Substitution of \eqref{eq3_1} into the first equation of the system \eqref{eq2_3} results in an equation $\sin(\theta_1-\theta_2)=0$, i.e. $\theta_1-\theta_2=0$ or $\theta_1-\theta_2=\pi$. Next, we restrict vortex intensities to positive values, i.e. we assume that $\Gamma_1>0$ and $\Gamma_2>0$.
In contrast to the case of the intensities with opposite signs \cite{SokRyabRCD2017}, in this particular situation the equation $\theta_1-\theta_2=0$ is impossible no matter which values parameters $c$ and $\varepsilon$ take.
In case of $\theta_1=\theta_2+\pi$, the second equation of the system \eqref{eq2_3} is reduced to
\begin{equation}
\label{eq3_2}
W(r_1,r_2)=0,
\end{equation}
where
\begin{equation*}
\begin{array}{l}
W(r_1,r_2)=(1-r_1^2)(1-r_2^2)\Bigl\{[c(1+r_1r_2)+\varepsilon](\Gamma_1r_1-\Gamma_2r_2)-\varepsilon(\Gamma_1r_1^3-\Gamma_2r_2^3)\Bigr\}-\\[3mm]
-r_1r_2(r_1+r_2)(1+r_1r_2)[\Gamma_1(1-r_2^2)-\Gamma_2(1-r_1^2)].
\end{array}
\end{equation*}

Substituting \eqref{eq3_1} into the Hamiltonian \eqref{eq1_2} and the vorticity moment \eqref{eq1_4} in the case where $\theta_1=\theta_2+\pi$, leads to the following values of the first integrals:
\begin{equation}
\label{eq3_3}
\begin{array}{l}
\displaystyle{h=\frac{1}{2}\{\Gamma_1^2\ln(1-r_1^2)+\Gamma_2^2\ln(1-r_2^2)\}+\Gamma_1\Gamma_2\ln\Bigl[\frac{(1+r_1r_2)^\varepsilon}{(r_1+r_2)^{c+\epsilon}}\Bigr],}\\[3mm]
f=\Gamma_1r_1^2+\Gamma_2r_2^2.
\end{array}
\end{equation}
This system \eqref{eq3_3} together with the equation \eqref{eq3_2} defines an implicit bifurcation diagram on the plane ${\mathbb R}^2(f,h)$.

\subsection{Special cases of bifurcation diagram parametrization}
In some special cases, it was possible to find an explicit parametrization of the bifurcation diagram.

Let
\begin{equation*}
\varepsilon=0.
\end{equation*}
After reduction by a non-zero factor of $1+r_1r_2$, the equation \eqref{eq3_2} takes the form
\begin{equation}
\label{eq4_6}
c(1-r_1^2)(1-r_2^2)(\Gamma_1r_1-\Gamma_2r_2)-r_1r_2(r_1+r_2)[\Gamma_1(1-r_2^2)-\Gamma_2(1-r_1^2)]=0.
\end{equation}
The algebraic curve \eqref{eq4_6} might be parametrized in the form of
\begin{equation*}
\label{x2_3}
\begin{array}{l}
\displaystyle{r_1=\frac{1}{\sqrt{2}}\sqrt{\frac{c(\Gamma_1-\Gamma_2t)(1+t^2)+(t+1)[(\Gamma_1-\Gamma_2)t\pm\sqrt{\cal D}]}
{t[\Gamma_1 t (c + t + t^2) - \Gamma_2 (1 + t + c t^2)]}}},\\[5mm]
r_2=t\cdot r_1, \text{\quad where\quad}{\cal D}=[c(1-t)(\Gamma_1-\Gamma_2t)+(\Gamma_1+\Gamma_2)t]^2-4\Gamma_1\Gamma_2t^2.
\end{array}
\end{equation*}
The corresponding bifurcation diagram $\Sigma$ is given as a curve on the plane ${\mathbb R}^2(f,h)$:
\begin{equation}
\label{x2_4}
\Sigma:\left\{
\begin{array}{l}
f=(\Gamma_1+\Gamma_2t^2)r_1^2,\\[3mm]
\displaystyle{h=\frac{1}{2}\Bigl\{\Gamma_1^2\ln(1-r_1^2)+\Gamma_2^2\ln(1-t^2r_1^2)-c\Gamma_1\Gamma_2\ln[(1+t)^2r_1^2]\Bigr\},}\\[3mm]
\displaystyle{r_1^2=\frac{c(\Gamma_1-\Gamma_2t)(1+t^2)+(t+1)[(\Gamma_1-\Gamma_2)t\pm\sqrt{\cal D}]}
{2t[\Gamma_1 t (c + t + t^2) - \Gamma_2 (1 + t + c t^2)]}.}

\end{array}\right.
\end{equation}

In the case of a vortex pair of positive intensities, i.e. $\Gamma_1=\Gamma_2=1$, after substitution
\eqref{eq3_1} and $\theta_1=\theta_2+\pi$ in \eqref{eq3_4} the critical set $\cal C$ also takes the simple form
\begin{equation}
\label{eq4_7}
\left\{\begin{array}{l}
\theta_1=\theta_2+\pi;\\
\left[\begin{array}{l}
r_1=r_2;\\
c(1 - r_1^2)(1 - r_2^2)-r_1r_2(r_1 + r_2)^2=0.
\end{array}\right.
\end{array}\right.
\end{equation}
The last equation of the system \eqref {eq4_7} coincides with the equation in the paper \cite {kevrek2013} on P.~225301-2 derived entirely from other considerations. Thus, our conclusion explains that the equation \cite{kevrek2013} on p.~225301-2 defines the radii of critical vortex motions.

In this case the corresponding bifurcation diagram $\Sigma$ consists of two curves $\gamma_1$ and $\gamma_2$, where
\begin{equation}\label{x2_5}
\begin{array}{l}
\displaystyle{\gamma_1: h=\ln\Bigl(1-\dfrac{f}{2}\Bigr)-\frac{c}{2}\ln(2f),\quad 0<f<2;}\\[3mm]
\gamma_2: \left\{
\begin{array}{l}
\displaystyle{h=\frac{1}{2}\ln\left[\frac{s^2(s-1)}{c+s-1}\right]-\frac{1}{2}c\ln\left[\frac{cs^2}{c+s-1}\right],}\\[3mm]
\displaystyle{f=\frac{cs^2-2(s-1)(c+s-1)}{c+s-1},}
\end{array}\right.\qquad s\in \left(1;\frac{2(1+\sqrt{c})}{2+\sqrt{c}}\right].
\end{array}
\end{equation}
For the values of the phy\-si\-cal pa\-ra\-me\-ter $c>3 $, the curve $\gamma_2$ has a cusp at  $s=\tfrac{\bigl[2-c+\sqrt{c(c-2)}\bigr](c-1)}{c-2}$, which coincides with the point of tangency, when $c=3$ and $s=\tfrac{2(1+\sqrt{c})}{2+\sqrt{c}}$.

The parameterized curve \eqref{x2_4} also have the cusps points that satisfy the equation
\begin{equation}
\label{eq4_8}
a_3c^3+a_2c^2+a_1c+a_0=0,
\end{equation}
where
\begin{equation*}
\begin{array}{l}
a_3=4\Gamma_1\Gamma_2(1-t)^2t^2(\Gamma_1-\Gamma_2t)^4,\\[3mm]
a_2=2(1-t)t(\Gamma_1-\Gamma_2t)^2\bigl[2(\Gamma_1+\Gamma_2)^3t^2(\Gamma_1-\Gamma_2t)-(\Gamma_1^2+\Gamma_2^2)(1+t)(\Gamma_1^2-\Gamma_2^2t^4)\bigr];\\[3mm]
a_1=-4\Gamma_2^6t^6(t^2+t-2)+4\Gamma_1^6t^2(2t^2-t-1)+\Gamma_1^5\Gamma_2(1+2t-3t^2+24t^4-16t^5)-\\[3mm]
-4\Gamma_1^2\Gamma_2^4t^3(-1-3t+10t^2-3t^3+t^5)+\Gamma_1\Gamma_2^5t^5(-16+24t-3t^3+2t^4+t^5)+\\[3mm]
+4\Gamma_1^4\Gamma_2^2t^2(-1+3t^2-10t^3+3t^4+t^5)-2\Gamma_1^3\Gamma_2^3t^2(2-11t^2+10t^3-11t^4+2t^6);\\[3mm]
a_0=-2(\Gamma_1-\Gamma_2)^2(\Gamma_1^2+\Gamma_2^2)t^3[-2\Gamma_1\Gamma_2t^2+\Gamma_2^2t^3(2+t)+\Gamma_1^2(1+2t)].
\end{array}
\end{equation*}

Moreover, the discriminant of the left side polynomial in the equation \eqref{eq4_8} describes a situation where the cusps points ``merge'' into one and one of the branches becomes smooth, which leads to a bifurcation diagram describing the interaction of two point vortices in an ideal fluid inside a circular cylinder \cite{BorMamSokolovskii2003}.

As an addition, we investigate the stability features of critical circles whose radii satisfy \eqref{eq4_7} and which lie
in the preimage of the bifurcation curves \eqref{x2_4} and \eqref{x2_5}. In this case, it is sufficient to determine the type
(elliptic/hyperbolic) in any one of the points $(f, h)$ on a smooth branch of the curve $\Sigma$ \cite{BolBorMam1}.

The type of a critical point $x_0$ with rank $1$ in an integrable system with two degrees of freedom can be calculated the following
way. One should specify the first integral $F$, such that $dF(x_0)=0$ and
$ dF \ne 0$ in a neighborhood of this point. The point $x_0$ is a
fixed point for the Hamiltonian vector field $\sgrad F$ and it is possible to calculate the linearization of this field at a given point -- the operator $A_F$ at the point $x_0$. This operator will have two zero eigenvalues and the remaining factor of the characteristic polynomial is $\mu^2-C_F$, where $C_F=\frac {1}{2}\trace(A_F ^ 2)$. When $C_F<0 $ we get the point of a type ``center'' (the corresponding periodic solution is elliptic, it is a stable periodic solution in phase space, the limit of the concentric family of two-dimensional regular tori), and for $C_F>0$ we get the point of a type ``saddle'' (the corresponding periodic solution is hyperbolic and there are motions, asymptotic to this solution, lying on two-dimensional separatrix surfaces). Here, explicit expressions for $C_F$ are presented only for bifurcation curves $\gamma_1$ and $\gamma_2$:

\begin{equation*}\label{x2_6}
\begin{array}{l}
\gamma_1: C_F=(4-c)f^2+4cf-4c,\quad 0<f<2;\\
\gamma_2: C_F=(c-2)s^2+2(c-1)(c-2)s-2(c-1)^2, \quad s\in \left(1;\frac{2(1+\sqrt{c})}{2+\sqrt{c}}\right].
\end{array}
\end{equation*}

Fig.~1 a), b) show an enlarged fragment of the bifurcation diagram in the case of the identical
intensities and $a=1$, while the parameter $c>3$ and the deformation parameter $\varepsilon=0$. The signs $``+``$ and $``-``$ correspond to elliptic (stable) and hyperbolic (unstable) periodic solutions in the phase space. As expected, the type change occurs at the cusp $A$ and the point of tangency $B$, both depicted on the bifurcation diagram $\Sigma$.

\begin{figure}[!ht]
\centering
\includegraphics[width=1\textwidth]{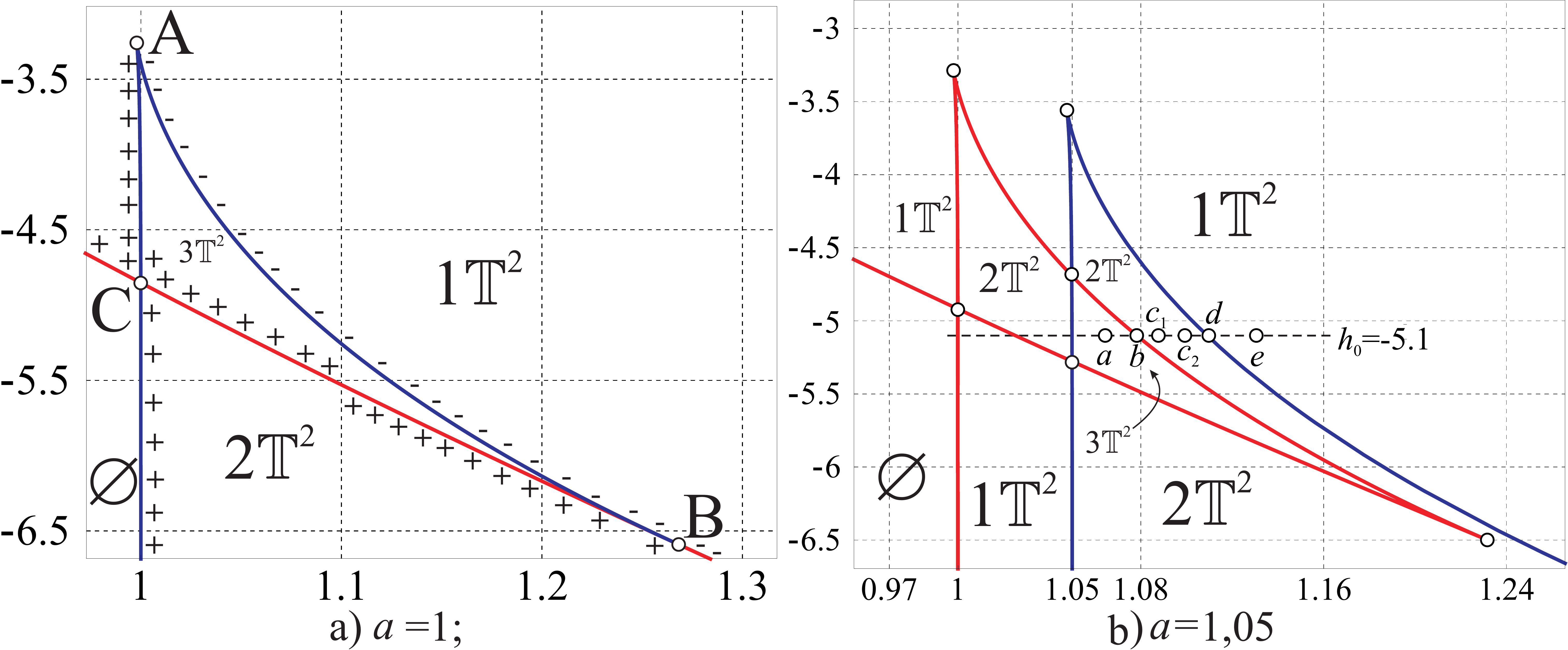}
\caption{a) Enlarged fragment of the bifurcation diagram $\Sigma$ where $\Gamma_1=\Gamma_2=1$ and $c>3,\varepsilon=0$; b) $\Sigma$-perturbation
where $\Gamma_1=1,\Gamma_2=1.05$ and $c>3, \varepsilon=0$.}
\label{fig1}
\end{figure}

Let
\begin{equation}
\label{eq4_11}
\Gamma_1=\Gamma_2=1, \varepsilon \ne 0.
\end{equation}
After substituting \eqref{eq3_1}, \eqref{eq4_11} and $\theta_1=\theta_2+\pi$ into \eqref{eq2_2}, the critical set $\cal C$ also takes a simple form
\begin{equation*}
\left\{\begin{array}{l}
\theta_1=\theta_2+\pi;\\
\left[\begin{array}{l}
r_1=r_2;\\
(1+r_1r_2)[(r_1^2+r_2^2)(r_1r_2+c)-(c-2)r_1^2r_2^2-c]+\\[3mm]
+\varepsilon(1-r_1^2)(1-r_2^2)(r_1^2+r_1r_2+r_2^2-1)=0.
\end{array}\right.
\end{array}\right.
\end{equation*}

The corresponding bifurcation diagram $\Sigma$ is defined on the plane
${\mathbb R}^2(f,h)$ and consists of two curves $\gamma_1$ and $\gamma_2$, where
\begin{equation*}
\begin{array}{l}
\displaystyle{\gamma_1: h=\ln\Bigl(1-\dfrac{f}{2}\Bigr)-\frac{1}{2}(c+\varepsilon)\ln(2f)+\varepsilon\ln\left(1+\frac{f}{2}\right),\quad 0<f<2;}\\[3mm]
\gamma_2: \left\{
\begin{array}{l}
\displaystyle{h=\ln\Bigl(x^\varepsilon\sqrt{x^2-1}z^{1-c}\Bigr),}\\[3mm]
\displaystyle{f=z^2-2xz+2,} \\[3mm]
\displaystyle{z=\frac{x[(\varepsilon+c)(x^2-1)+1]}{(\varepsilon+1)x^2-\varepsilon}},
\end{array}\right. x\in (1; x_0].
\end{array}
\end{equation*}
Here $ x_0$ denotes the root of the equation
\begin{equation}\label{y5}
  (z-2x)^2=4(x^2-1),\quad  x>1.
\end{equation}

Fig.~2 and 3 show the bifurcation diagram and its enlarged fragment in the case of \eqref{eq4_11} for the parameter values $\varepsilon =28, c=12$. Note that the curve $\gamma_2$ has the cusps points $A, B$, and $C$ is the point of tangency between $\gamma_2$ and $\gamma_1$ for the specified parameter values, where
\begin{equation*}
\begin{array}{l}
\displaystyle{x_0=\frac{\sqrt{570}}{1140}\sqrt{2312+\sqrt[3]{2885048-294690\sqrt{6}}+\sqrt[3]{2885048+294690\sqrt{6}}}\approx1,06678;}\\[3mm]
\displaystyle{f_C=\frac{2}{57}\Bigl[45+\frac{\sqrt[3]{141^2}}{\sqrt[3]{45+19\sqrt{6}}}-\sqrt[3]{141(45+19\sqrt{6})}\Bigr]\approx 0,9667958154;}\\[3mm]
h_C\approx -2,8066772742.
\end{array}
\end{equation*}
Shown on Fig.~3~a), the tangent point $C$ satisfies \eqref{y5}.

\begin{figure}[!ht]
\centering
\includegraphics[width=0.5\textwidth]{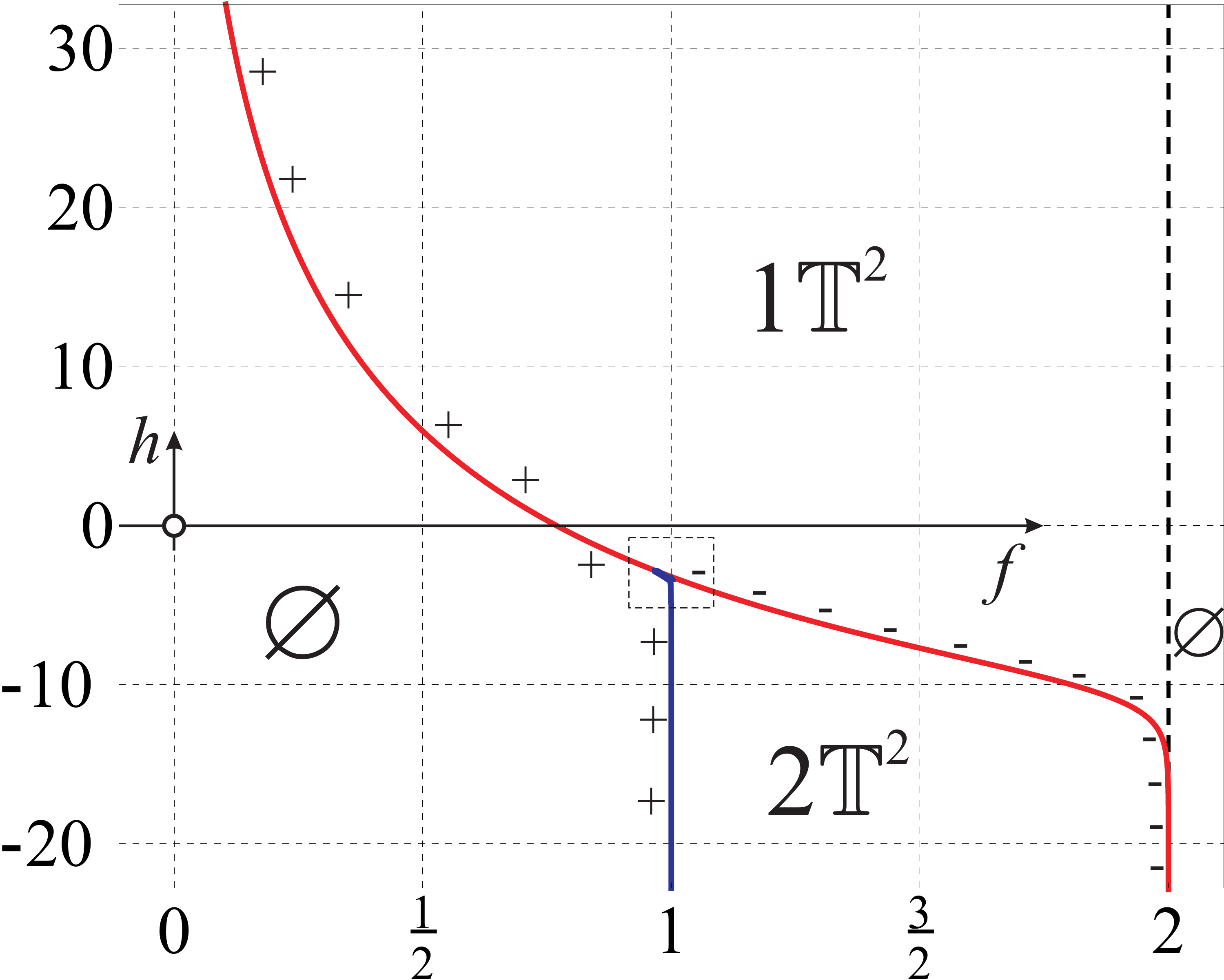}
\caption{Bifurcation diagram $\Sigma$ with $\Gamma_1=\Gamma_2=1$ and $c=12, \varepsilon=28$.}
\label{fig2}
\end{figure}
\begin{figure}[!ht]
\centering
\includegraphics[width=1\textwidth]{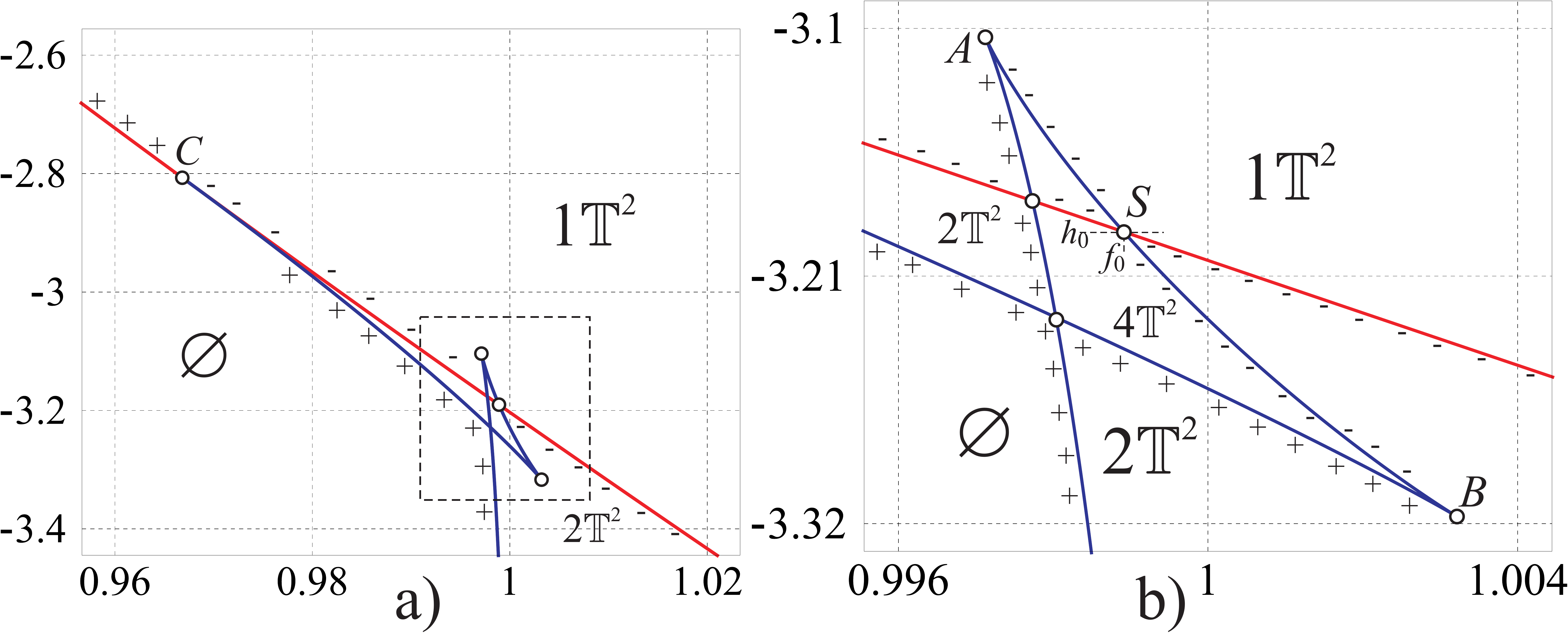}
\caption{Enlarged fragment of the bifurcation diagram with $\Gamma_1=\Gamma_2=1$ $c=12;\varepsilon=28$.}
\label{fig3}
\end{figure}

Again, the signs  $``+``$ and $``-``$ correspond to elliptic (stable) and hyperbolic periodic solutions in phase space. The type change occurs at the cusp points $A$ and $B$, as well as at the point of tangency $C$, all shown on the bifurcation diagram $\Sigma$.
For clarity, an explicit expression of the coefficient $C_F$ is shown. It is responsible for the type
(elliptic/hyperbolic) of the smooth branch of the curve $\gamma_1$:

\begin{equation*}\label{y_6}
\begin{array}{l}
\gamma_1: C_F=(4-c+3\varepsilon)f^3+2(c+4-7\varepsilon)f^2+4(c+5\varepsilon)f-8(c+\varepsilon),
\quad 0<f<2.
\end{array}
\end{equation*}
At $C_F < 0$ we get the point of a type ``center'', and at $C_F > 0$ we get the point of type ``saddle''.

\subsection{General case of bifurcation diagram's parametrization}

Here we give a fragment of the bifurcation diagram $\Sigma$ (Fig.~4) and its dynamics (Fig.~5) for the general case, using implicit parametrization in the form of the equation \eqref{eq3_2} and dependencies \eqref{eq3_3}. The program of an interactive bifurcation diagram visualization was written in the \textit{Python} programming language, using interactive environment of \textit{Jupyter Notebook}. Having specified values of the parameters $\Gamma_1=1,\Gamma_2=1,0015$ $c=12,\varepsilon=28$, numerical methods were implemented in order to solve the polynomial equations \eqref{eq3_2} for $r_1$, given $r_2\in(0;1)$. Thus on a plane of coordinates $(f,h)$ the bifurcation diagram $\Sigma$ was plotted in the form of dependencies \eqref{eq3_3}. When the physical parameter of the intensity ratio is perturbed (for clarity, $a=\frac{\Gamma_2}{\Gamma_1}=1,0015$), the ``separation'' of the point $C$ is observed (on Fig.~3(a)  it corresponds to the point of tangency $C$). This leads to the perturbation of a part of the bifurcation diagram. This situation is typical for the perturbation of bifurcation diagrams of integrable systems that have points of tangency between bifurcation sheets. For example, such pattern holds for bifurcation diagrams of Kovalevskaya's top integrable cases and its generalization to Kovalevskaya-Yehia gyrostat in a rigid body dynamics \cite{KhRyab2017JMS}.

\begin{figure}[!ht]
\centering
\includegraphics[width=1\textwidth]{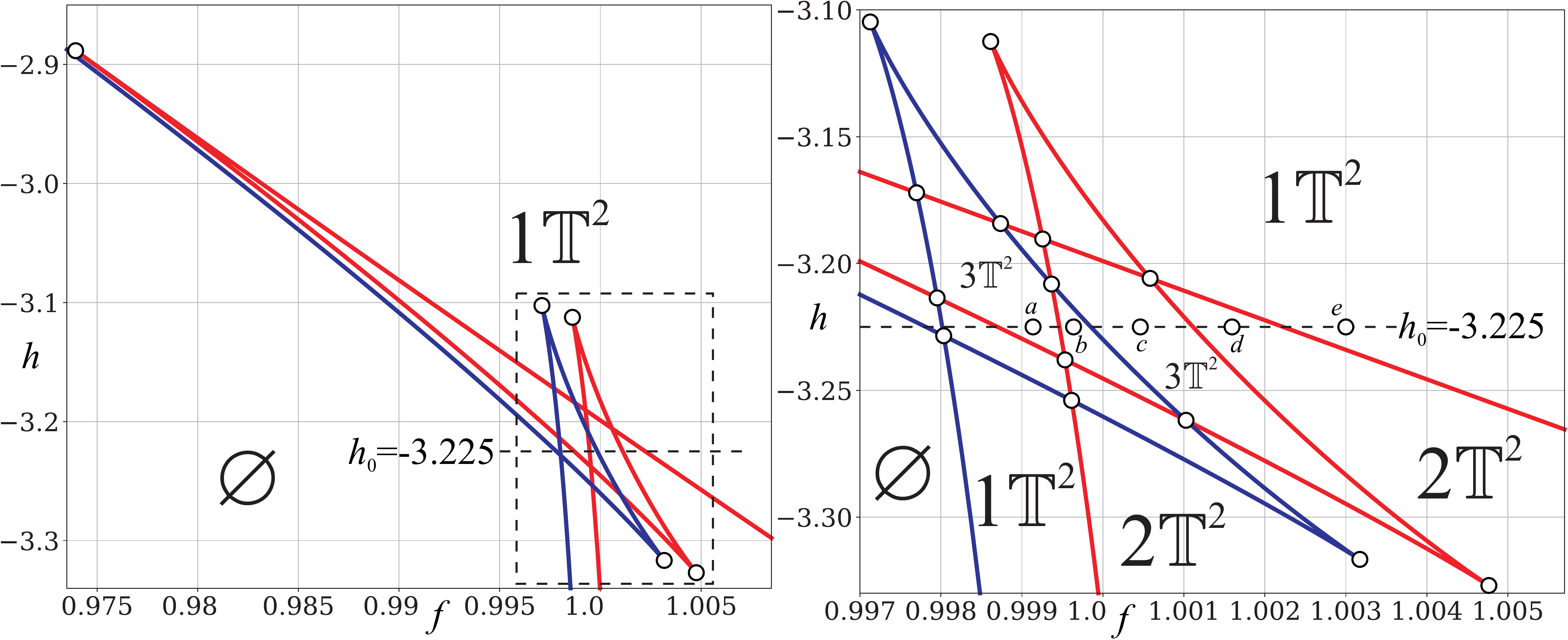}
\caption{Enlarged fragment of the bifurcation diagram perturbation for $\Gamma_1=1;\Gamma_2=1,0015$ $c=12;\varepsilon=28$.}
\label{fig4}
\end{figure}

Fig.~5 shows fragments of the bifurcation diagram's change in dynamics using interactive visualization program described above. Along with the change of parameters, such as the ratio of intensities $a=\frac{\Gamma_2}{\Gamma_1}$, the interaction parameter of vortices $c$ and the deformation parameter $\varepsilon$, the formation of triangular regions, bounded by pieces of bifurcation curves, is observed. These regions experience various deformations (e.g., some triangular regions disappear) and finally we can observe a stable bifurcation diagram which corresponds to the problem of dynamics of two vortices bounded by a circular region in an ideal fluid \cite{BorMamSokolovskii2003}.

\begin{figure}[!ht]
\centering
\includegraphics[width=1\textwidth]{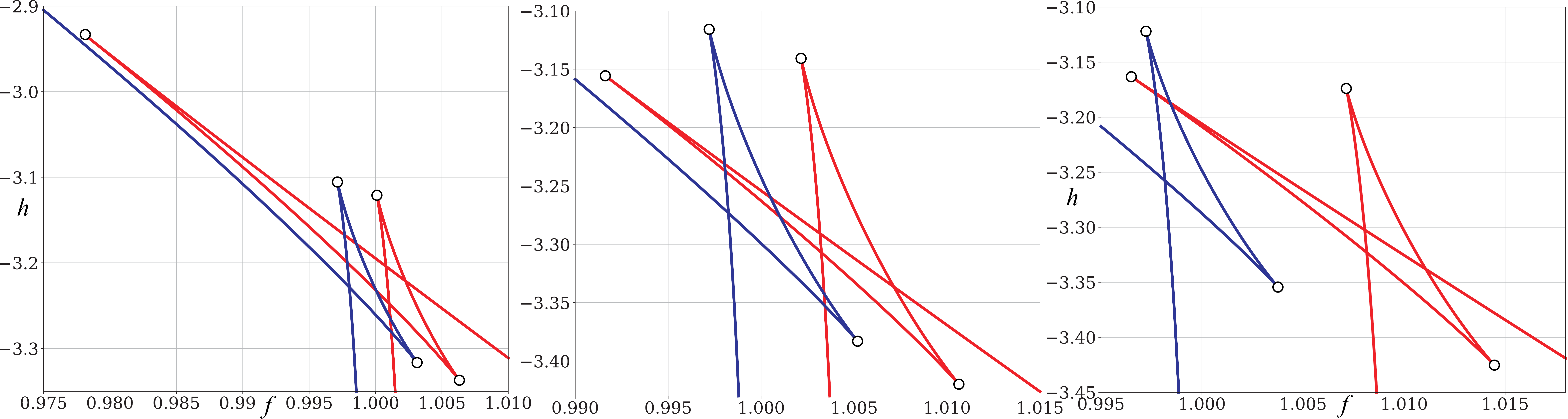}\newline
\includegraphics[width=1\textwidth]{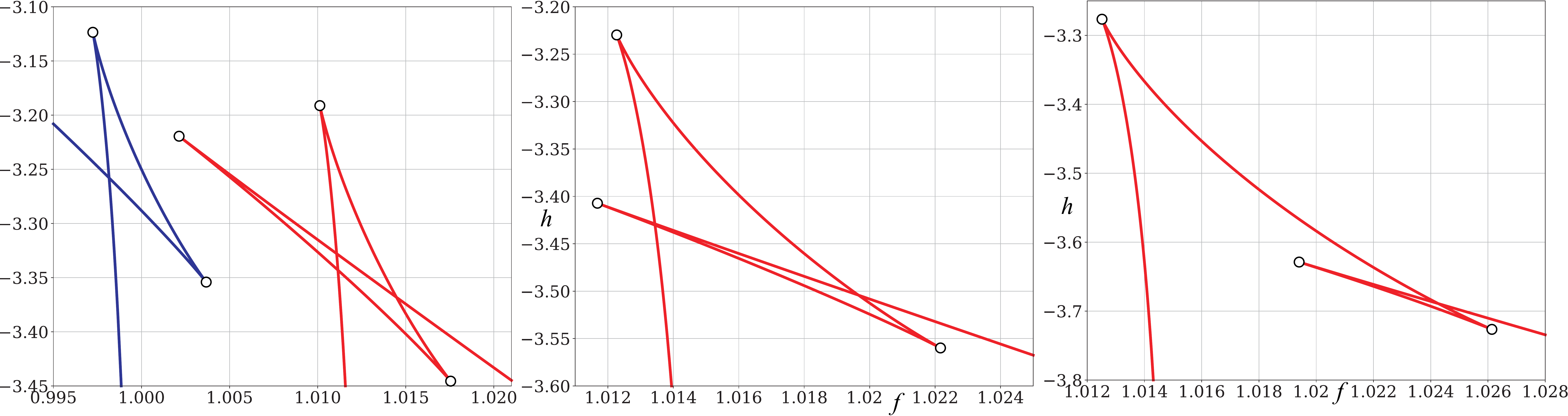}\newline
\includegraphics[width=1\textwidth]{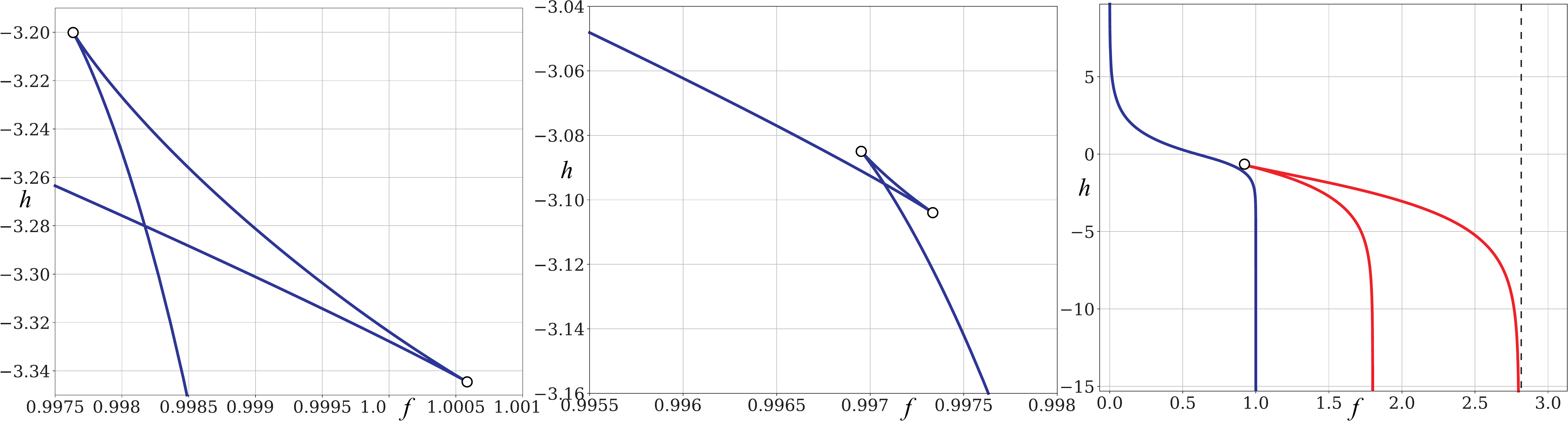}
\caption{Dynamics of the bifurcation diagram $\Sigma$ in the general case.}
\label{fig5}
\end{figure}

\section{REDUCTION TO A SYSTEM WITH ONE DEGREE OF FREEDOM}
Here we restrict intensities' values by positive ones, i.e. throughout this section we assume that the intensity parameters $\Gamma_1$ and $\Gamma_2$ have a positive sign. Let's perform an explicit reduction to a system with one degree of freedom. In order to perform this for the system \eqref{eq1_1} with Hamiltonian \eqref{eq1_2}, one should substitute phase variables $(x_k, y_k)$ to new variables $(u,v,\alpha)$ using the formulas below
\begin{equation*}\label{z1}
\begin{array}{l}
\displaystyle{x_1=\frac{1}{\sqrt{\Gamma_1}}[u\cos(\alpha)-v\sin(\alpha)],\quad y_1=\frac{1}{\sqrt{\Gamma_1}}[u\sin(\alpha)+v\cos(\alpha)],}\\[3mm]
\displaystyle{x_2=\frac{1}{\sqrt{\Gamma_2}}\sqrt{f-u^2-v^2}\cos(\alpha),\quad y_2=\frac{1}{\sqrt{\Gamma_2}}\sqrt{f-u^2-v^2}\sin(\alpha).}
\end{array}
\end{equation*}

The physical variables $(u,v)$ are Cartesian coordinates of one of the vortices in a coordinate system that is associated with another vortex rotating around the center of vorticity. The choice of such variables is suggested by the presence of the integral of the angular momentum of vorticity \eqref{eq1_4}, which is invariant under the rotation group $SO(2)$. The existence of a one-parameter symmetry group allows to perform a reduction to a system with one degree of freedom in a similar fashion as in mechanical systems with symmetry \cite{Kharlamov1988}. Backward substitution

\begin{equation*}
U=\sqrt{\Gamma_1}\frac{x_1x_2+y_1y_2}{\sqrt{x_2^2+y_2^2}},\quad V=\sqrt{\Gamma_1}\frac{y_1x_2-x_1y_2}{\sqrt{x_2^2+y_2^2}},
\end{equation*}
leads to canonical variables with respect to the bracket \eqref{eq1_3}:
\begin{equation*}
\{U,V\}=-\{V,U\}=1,\quad \{U,U\}=\{V,V\}=0.
\end{equation*}

The system with respect to the variables $(u,v)$ is Hamiltonian:
\begin{equation}\label{z2}
\dot u=\frac{\partial H_1}{\partial v},\quad
\dot v=-\frac{\partial H_1}{\partial u},
\end{equation}
with Hamiltonian
\begin{equation}
\label{z3}
\begin{array}{l}
\displaystyle{H_1=
\frac{1}{2}\Bigl\{\Gamma_1^2\ln\Bigl(1-\frac{u^2+v^2}{\Gamma_1}\Bigr)+\Gamma_2^2\ln\Bigl(1-\frac{f}{\Gamma_2}+\frac{u^2+v^2}{\Gamma_2}\Bigr)-}\\[3mm]
\displaystyle{-\Gamma_1\Gamma_2(c+\varepsilon)\ln\Bigl[\Bigl(\frac{u}{\sqrt{\Gamma_1}}-\frac{\sqrt{f-u^2-v^2}}{\sqrt{\Gamma_2}}\Bigr)^2+\frac{v^2}{\Gamma_1}\Bigr]-}\\[3mm]
\displaystyle{+\varepsilon\Gamma_1\Gamma_2\ln\Bigl[\Bigl(1-\frac{u\sqrt{f-u^2-v^2}}{\sqrt{\Gamma_1\Gamma_2}}\Bigr)^2+\frac{v^2(f-u^2-v^2)}{\Gamma_1\Gamma_2}\Bigr].\Bigr\}}

\end{array}
\end{equation}

The rotation angle $\alpha(t)$ of the rotating coordinate system satisfies the differential equation
\begin{equation*}
\begin{array}{l}
\displaystyle{\dot\alpha=
\frac{\Gamma_2^2}{\Gamma_2-f+u^2+v^2}+c\Gamma_2\Gamma_1\sqrt{\Gamma_1}\frac{R_1(u,v)}{Q_1(u,v)}+\varepsilon\Gamma_1\Gamma_2\sqrt{\Gamma_2}
\frac{(\Gamma_1-u^2-v^2)}{\sqrt{f-u^2-v^2}}\frac{R_2(u,v)}{Q_2(u,v)}},
\end{array}
\end{equation*}
where
\begin{equation*}
\begin{array}{l}
R_1(u,v)=\Gamma_1(f-u^2-v^2)-u^2\Gamma_2,\\[3mm]
Q_1(u,v)=\sqrt{\Gamma_2}u\sqrt{f-u^2-v^2}\bigl[\Gamma_2(u^2+v^2)-\Gamma_1(f-u^2-v^2)\bigr]+\\[3mm]
+\sqrt{\Gamma_1}(f-u^2-v^2)\bigl[\Gamma_1(f-u^2-v^2)-\Gamma_2(u^2-v^2)\bigr]
\end{array}
\end{equation*}
\begin{equation*}
\begin{array}{l}
R_2(u,v)=\sqrt{\Gamma_2}\sqrt{f-u^2-v^2}(\Gamma_1+u^2+v^2)-\sqrt{\Gamma_1}u(\Gamma_2+f-u^2-v^2),\\[3mm]
Q_2(u,v)=\Gamma_2\bigl[\Gamma_1(\Gamma_1+4u^2)+(u^2+v^2)^2\bigr](f-u^2-v^2)+\Gamma_1(u^2+v^2)\bigl[\Gamma_2^2+(f-u^2-v^2)^2\bigr]-\\[3mm]
-2u\sqrt{\Gamma_1\Gamma_2}\sqrt{f-u^2-v^2}(\Gamma_1+u^2+v^2)(\Gamma_2+f-u^2-v^2).
\end{array}
\end{equation*}

The fixed points of the reduced system \eqref{z2} are determined by the critical points of the reduced Hamiltonian \eqref{z3} and correspond to the relative equilibria of vortices in the system \eqref{eq1_1}. For a fixed value of an integral of the moment of vorticity $f$, the regular levels of the reduced Hamiltonian are compact and motions occur along closed curves. It can be shown that the critical values of the reduced Hamiltonian determine the bifurcation diagram \eqref{eq3_2}, \eqref{eq3_3}.

Note some interesting special cases.
For a segment of the bifurcation curve $(AB)$ (Fig.~\ref{fig1}), the motion on the plane $(u,v)$ occurs along a curve that is topologically structured as $\mathbb S^1\,\dot{\cup}\,\mathbb S^1\,\dot{\cup}\,\mathbb S^1$ (Fig.~6b)), and the integral critical surface is a trivial bundle over $\mathbb S^1$ with the layer $\mathbb S^1\,\dot{\cup}\,\mathbb S^1\,\dot{\cup}\,\mathbb S^1$.

When passing through a section of the curve $(AB)$ in the case of $c>3$ (Fig.~\ref{fig1}), the bifurcation of three tori into one occurs as follows $3\mathbb T^2 \to \mathbb S^1\times\left(\mathbb S^1\,\dot{\cup}\,\mathbb S^1\,\dot{\cup}\,\mathbb S^1\right)\to \mathbb T^2$. With the help of the reduced Hamiltonian level curves in Fig.~\ref{fig6}, this bifurcation is clearly visible ($h_1=-4.5$ for a) $f_1=1.04$; b) $f_{2}=1.042957$; c) $f_3=1.05$).

For another special case of the bifurcation diagram (Fig.~3 b), where the bifurcation curves $\gamma_1$ and $\gamma_2$ intersect at the point $S$ ($x_S=1,008383; f_0=0,9989101; h_0=-3,1903429$), the movement occures on the plane $(u,v)$ and goes by the curve which is topologically arranged as $\mathbb S^1\,\dot{\cup}\,\mathbb S^1\,\dot{\cup}\,\mathbb S^1\,\dot{\cup}\,\mathbb S^1$ (Fig.~7~a). In this case the integral critical surface is a trivial bundle over $\mathbb S^1$ with a layer of $\mathbb S^1\,\dot{\cup}\,\mathbb S^1\,\dot{\cup}\,\mathbb S^1\,$ $\dot{\cup}\,\mathbb S^1$.
During cross-over of the point $S$, moving along the line $h=h_0$ on the bifurcation diagram $\Sigma$, the bifurcation of four tori into one occures $4\mathbb T^2 \to \mathbb S^1\times\left(\mathbb S^1\,\dot{\cup}\,\mathbb S^1\,\dot{\cup}\,\mathbb S^1\,\dot{\cup}\,\mathbb S^1\right)$ $\to \mathbb T^2$.
With the help of the reduced Hamiltonian \eqref{z3} level curves, this four-into-one bifurcation is clearly seen on Fig.~7.
The following parameters were used here: $\Gamma_1=\Gamma_2=1,\; c=12,\; \varepsilon=28$.

\begin{figure}[!ht]
\centering
\includegraphics[width=1\textwidth]{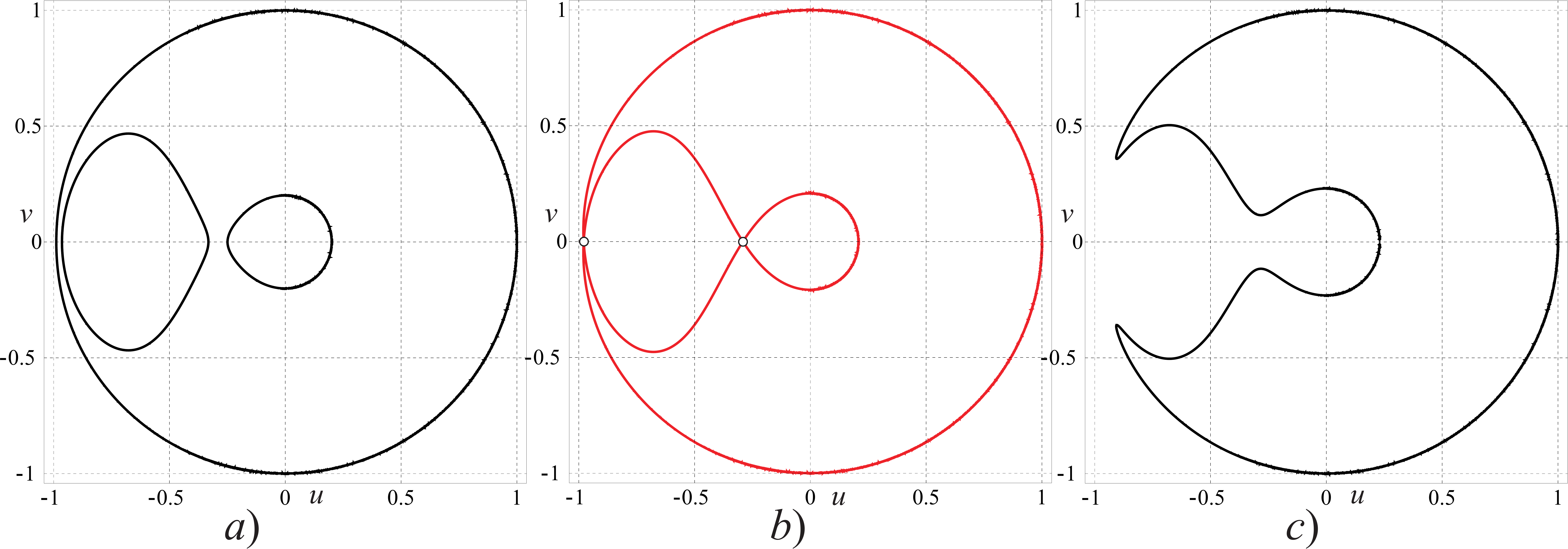}
\caption{Level curves of reduced Hamiltonian $H_1$ for $\Gamma_1=\Gamma_2=1$ and $c>3, \varepsilon=0$.}
\label{fig6}
\end{figure}

\begin{figure}[!ht]
\centering
\includegraphics[width=1\textwidth]{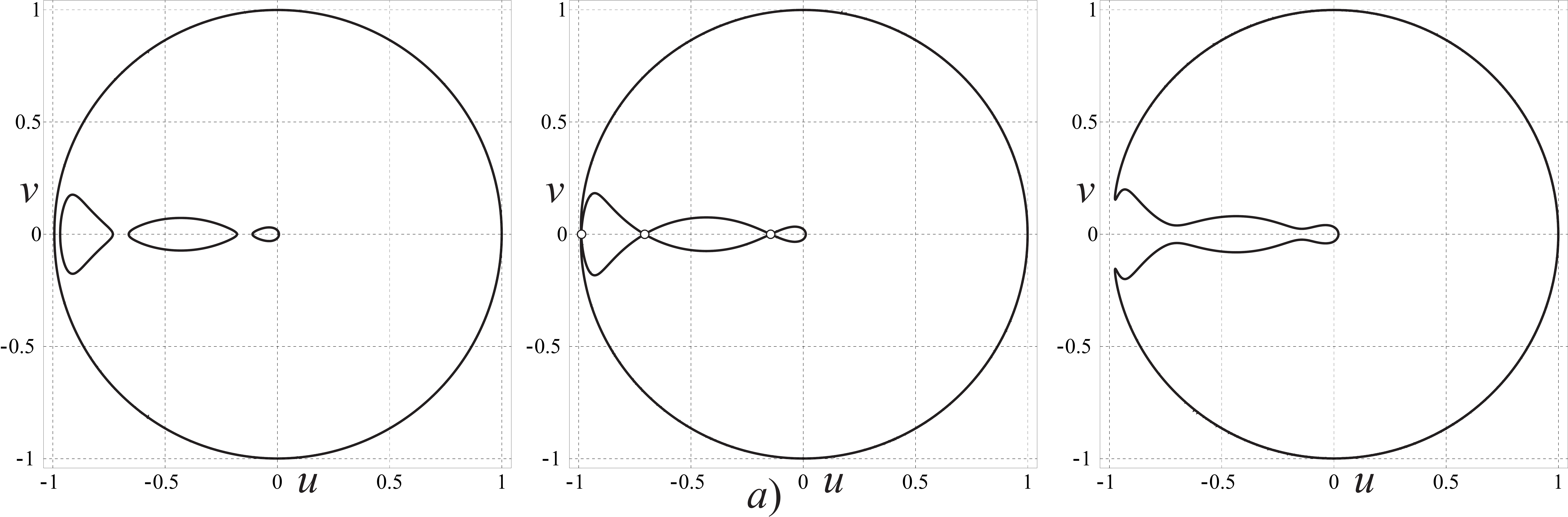}
\caption{Level curves of reduced Hamiltonian $H_1$ along the line $h=h_0$ for $\Gamma_1=\Gamma_2=1$ and $c=12, \varepsilon=28$.}
\label{fig7}
\end{figure}

Let's consider an integrable perturbation of the physical parameter of the intesities' ratio $a=\frac{\Gamma_2}{\Gamma_1}$ when the deformation parameter $\varepsilon$ is absent, i.e. equal to zero. In this case, a perturbation of a special layer of Liouville foliation is observed, i.e. in the terminology of singularities, the atom $D_1$ desintegrates into two atoms of type $B$ \cite{oshtuzh2018}. Fig.~8 corresponds to the section of the bifurcation diagram on Fig.~1 b) along the line $h=h_0=-5.1$. Here, for specified values of second integral's parameter $f=a\text{\,--\,}e$, the perturbation is clearly presented (Fig.~8 (b) (d) of the said feature (see~Fig.~6 (b))).
This result is also confirmed by \cite{oshtuzh2018}.

\begin{figure}[!ht]
\centering
\includegraphics[width=1\textwidth]{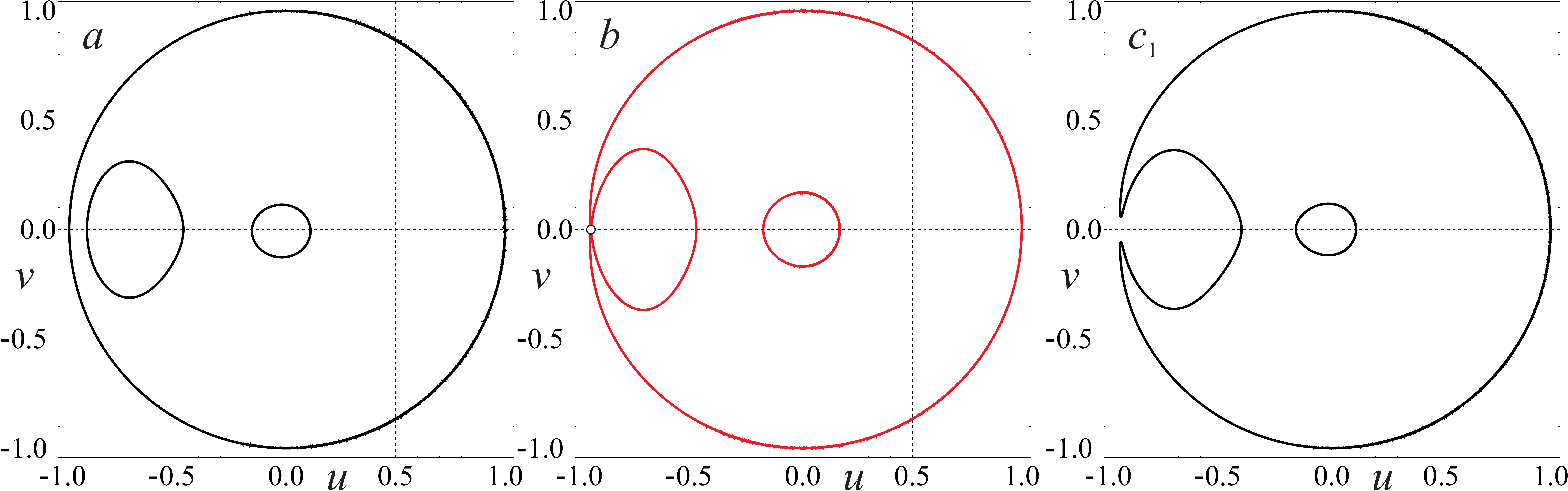}\newline
\includegraphics[width=1\textwidth]{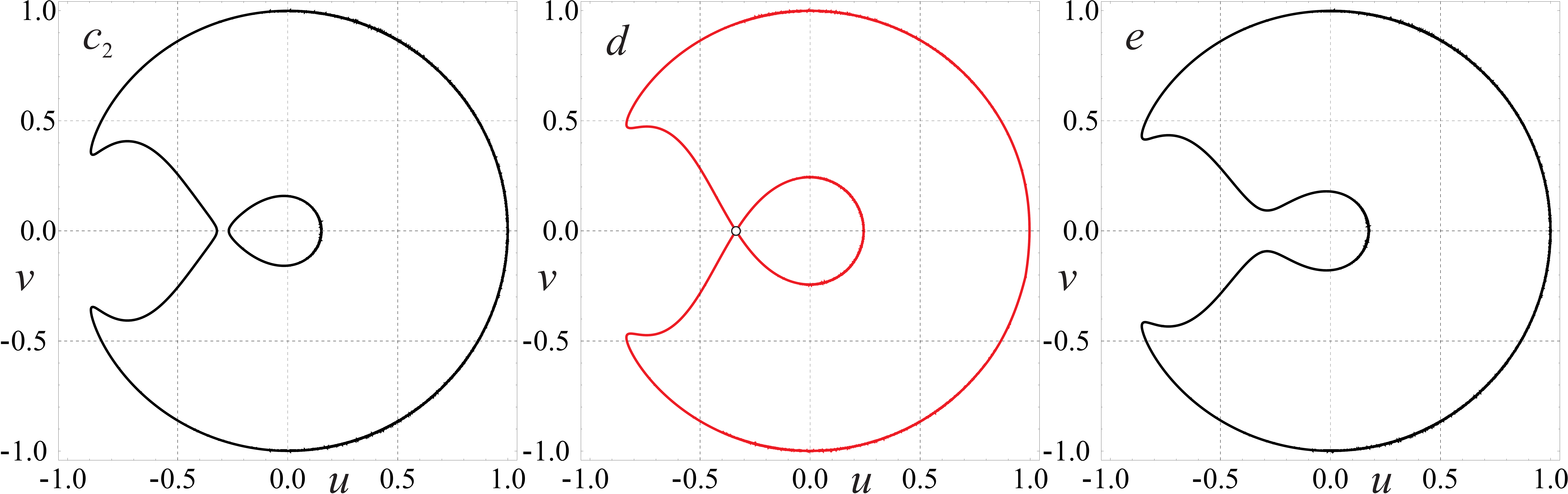}
\caption{Level curves of reduced Hamiltonian $H_1$ for $\Gamma_1=1;\Gamma_2=1,05$ and $c>3$.}
\label{fig8}
\end{figure}

Finally, for an interactive fragment of the bifurcation diagram perturbation (Fig.~4) for $\Gamma_1=1;\Gamma_2=1,0015$ $c=12;\varepsilon=28$ (here the parameters are arbitrary), the corresponding contour lines of the reduced Hamiltonian $H_1$ for the selected values of the second integral $f=a\text{\,--\,}e$ are shown on Fig.~9.

\begin{figure}[!ht]
\centering
\includegraphics[width=1\textwidth]{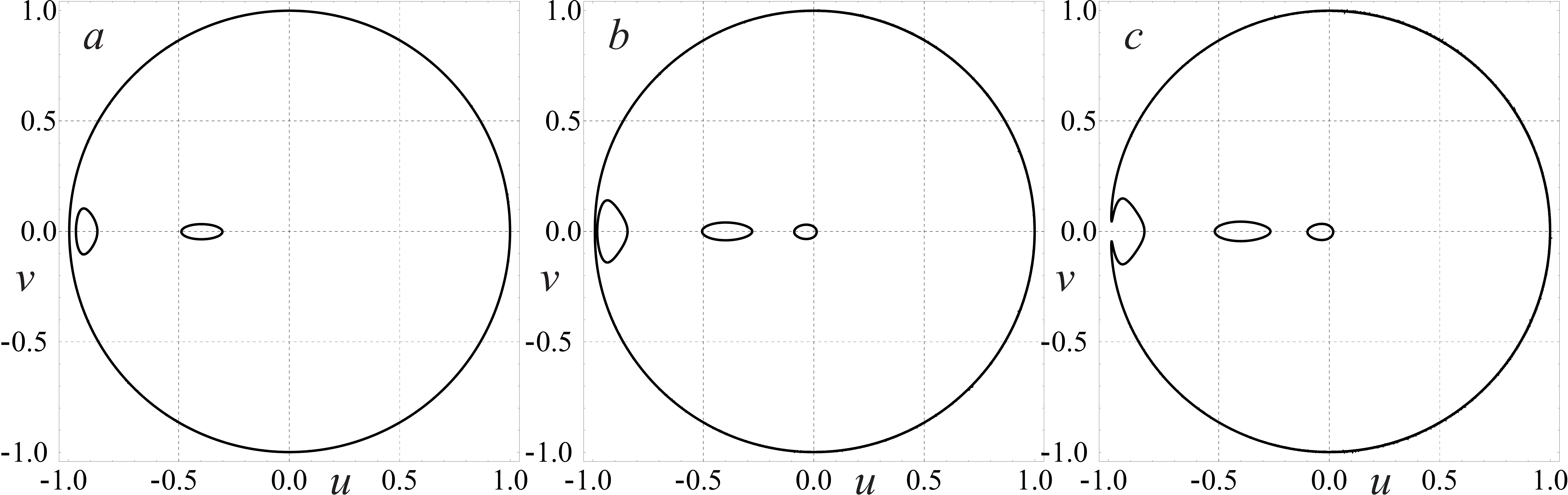}\newline
\includegraphics[width=1\textwidth]{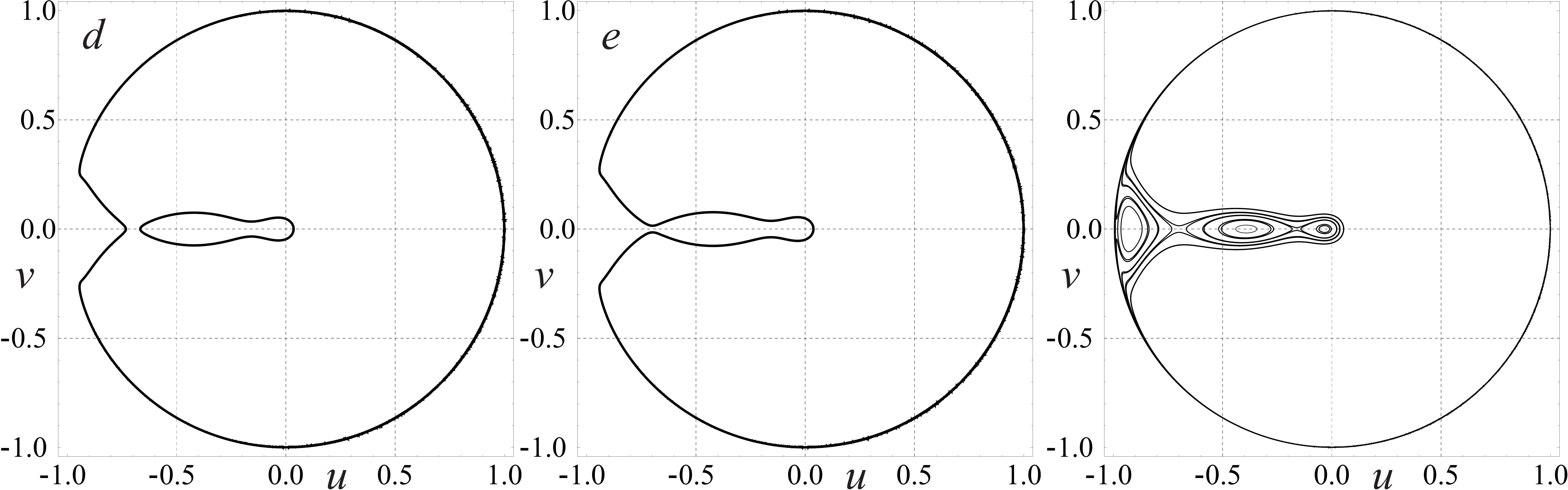}
\caption{Level curves of reduced Hamiltonian $H_1$ for $\Gamma_1=1; \Gamma_2=1.0015$ and $c=12,\varepsilon=28$ along the line $h=h_0-3.225$ at $f=a\text{\,--\,}e$.}
\label{fig9}
\end{figure}

In conclusion it should be mentioned, that implemented for \textit{arbitrary} intensities $\Gamma_1, \Gamma_2$, the physical parameter $c$ and the deformation parameter $\varepsilon$, computer model of absolute dynamics of the vortices in the fixed coordinate system described by \eqref{eq1_1} and \eqref{eq1_2} is based on the analytical results of this publication (an implicit parametrization of the bifurcation diagrams \eqref{eq3_2}, \eqref{eq3_3}, the reduction to the system with one degree of freedom \eqref{z2}, the stability analysis).
A classification of bifurcation diagrams (study of points of tangency bifurcations, cusps points, depending on the values of physical parameters $a, c, \varepsilon$) might be a separate research topic with the aim of creating a topological atlas for the generalized integrable model of vortex dynamics considered here.

\section*{ACKNOWLEDGMENTS}
The authors thank to Prof. A.\,V.\,Borisov for useful discussions.

\section*{FUNDING}
The work of P.\,E.\,Ryabov was supported by RFBR grant
17-01-00846 and was carried out within the framework of the state assignment of the
Ministry of Education and Science of Russia (project no.\, 1.2404.2017/4.6).

\section*{CONFLICT OF INTEREST}

The authors declare that they have no conflicts of interest.

\section*{AUTHORS’ CONTRIBUTIONS}
Interactive visualization of the bifurcation diagram is made by A.\,A.\,Shadrin based on the equations of a bifurcation diagram \eqref{eq3_2}, \eqref{eq3_3} and reduction to the system with one degree of freedom in general case \eqref{z2}.
Authors P.\,E.\,Ryabov and A.\,A.\,Shadrin were involved in writing the text of the paper. All authors participated in discussing the results.


\begin{thebibliography}{99}
\bibitem{fett2009} Fetter, A.\,L.,  Rotating trapped Bose-Einstein condensates, \textit{Rev.~Mod.~Phys.}, 2009, vol.\,81, no.\,2,  pp.\,647--691.


\bibitem{kevrekPhysLett2011} Torres, P.\,J., Kevre\-kidis, P.\,G., Frantzeska\-kis, D.\,J., Carre\-te\-ro-Gon\-zalez, \,R., Schmel\-cher, P. and Hall, D.\,S., Dynamics of vortex dipoles in confined Bose--Ein\-stein con\-den\-sates, \textit{Phys.~Lett.~A.},  2011, vol.\,375, pp.\,3044--3050.

\bibitem{kevrek2013} Navarro, R., Carretero-Gonz\'{a}lez, R., Torres, P.\,J.,  Kevrekidis, P.\,G.,
Frantzeskakis, D.\,J.,  Ray, M.\,W.,  Altunta\c{s}, E. and Hall, D.\,S., Dynamics of Few Co-rotating Vortices in Bose-Einstein Condensates, \textit{Phys. Rev. Lett.}, 2013, vol.\,110, no.\,22, pp.~225301-1--6.


\bibitem{kevrikidis2014} Koukouloyannis, \,V. and Voyatzis, \,G. and Kevrekidis, \,P.\,G., Dynamics of three noncorotating vortices in Bose--Einstein condensates, \textit{Phys. Rev. E.}, 2014, vol.\,89, no.\,4, pp.\,042905-1--14.

\bibitem{Greenhill1877} Greenhill, A.\,G., Plane vortex motion~// \textit{Quart. J. Pure Appl. Math.}, 1877/78, vol.~15, no.~58, pp.~10--27

\bibitem{BorMamSokolovskii2003} Kilin, A.\,A., Borisov, A.\,V. and Mamaev, I.\,S., The Dynamics of Point Vortices Inside and Outside
a Circular Domain, in \textit{Basic and Applied Problems of the Theory of Vortices}  Borisov,  A.\,V.  and Mamaev,  I.\,S. and Sokolovskiy, M.\,A. (Eds.), Izhevsk: Regular and Chaotic Dynamics, Institute of Computer Science, 2003, pp. 414-440
(Russian).

\bibitem{BorMam2005book} Kilin, A.\,A., Borisov, A.\,V., Mamaev~I.\,S.,
 The Dynamics of Point Vortices Inside and Outside a Circular Domain, in \textit{
 Mathematical methods of vortex structure dynamics} Borisov, A.\,V. and Mamaev, I.\,S.(Eds.)  M.-Izhevsk:
Regular and Chaotic Dynamics,, Institute of Computer Science, 2005, pp.~148--173 (Russian).

\bibitem{borkil2000} Borisov, А.\,V. and Kilin, A.\,A., Stability of Thomson's Configurations of Vortices on a Sphere, \textit{Regular and Chaotic Dynamics}, 2000, vol.\,5, no.\,2, pp.\,189--200.

\bibitem{bormamkil2004} Borisov, A.\,V.,  Mamaev, I.\,S. and Kilin, A.\,A., Absolute and relative choreographies in the problem of point vortices moving on a plane, \textit{Regular and Chaotic Dynamics}, 2004, vol.\,9, no.\,2, pp.\,101--111.

\bibitem{kilinbormam2013} Borisov, A.\,V.,  Kilin, A.\,A. and Mamaev, I.\,S., The Dyna\-mics of Vortex Rings: Leapf\-rog\-ging, Cho\-reographies and the Stability Problem, \textit{Regular and Chaotic Dynamics}, 2013, vol.\,18, nod.\,1--2, pp.\,33--62.

\bibitem{BorSokRyab2016} Borisov, A.\,V.,  Ryabov, P.\,E.  and Sokolov, S.\,V.,	
 Bifurcation analysis of the motion of a cylinder and a point vortex in an ideal fluid, \textit{Mathematical Notes}, 2016, vol.\,99, no.\,6,  pp.\,834--839.	

\bibitem{SokRyabRCD2017}
 Sokolov, S.\,V. and Ryabov, P.\,E., Bifurcation Analysis of the Dynamics of Two Vortices in a Bose–Einstein Condensate. The Case of Intensities of Opposite Signs, \textit{Regular and Chaotic Dynamics}, 2017, vol.\,22, no.\,8, pp.\,979--998.

\bibitem{sokryab2018} Sokolov, S.\,V. and Ryabov, P.\,E. Bifurcation Diagram of the Two Vortices in a Bose-Einstein Condensate with Intensities of the Same Signs, \textit{Doklady Mathematics}, 2018, vol.\,97, no.\,3, pp.\,1--5.

\bibitem{RyabDan2019} Ryabov, P.\,E. Bifurcations of Liouville Tori in a System of Two Vortices
of Positive Intensity in a Bose-Einstein Condensate,  \textit{Doklady Mathematics}, 2019, vol.\,99, no.\,2, pp.\,1--5.


\bibitem{RyabSocND2019} Ryabov, P.\,E., Sokolov, S.\,V. Phase Topology of Two Vortices of Identical Intensities in a Bose-Einstein Condensate, \textit{Rus. J. Nonlin. Dyn.}, 2019, vol.\,15, no.\,1, pp.\,59--66.

\bibitem{Kharlamov1988} Kharlamov, M.\,P., \textit{Topological Analysis of Integrable Problems of Rigid Body Dynamics}, Leningrad:
Leningr. Gos. Univ., 1988 (Russian).

\bibitem{bolsmatvfom1990} Bolsinov, A.\,V., Matveev, S.\,V. and Fomenko, A.\,T., Topological classification of integrable Hamiltonian systems with two degrees of freedom. List of systems of small complexity, \textit{Russian Mathematical Surveys}, 1990, vol.\,45, no.\,2, pp.\,59--94.

\bibitem{oshtuzh2018} Oshemkov,  A.\,A. and Tuzhilin, M.\,A., Integrable perturbations of saddle singularities of rank 0 of integrable Hamiltonian systems, \textit{Sbornik: Mathematics}, 2018, vol.\,209, no.\,9, pp.\,1351--1375.	


\bibitem{RyabovArXiv2019}
Ryabov\,P.\,E. On bifurcation of the four Liouville tori in one generalized integrable model of the vortex dynamics,
\texttt{https://arxiv.org/abs/1903.09945}\,(Russian).

\bibitem{BolBorMam1} Bolsinov, A.\,V.,  Borisov, A.\,V.  and Mamaev, I.\,S., Topology and Stability of Integrable Systems, \textit{Russian
Math. Surveys}, 2010, vol.\,65, no.\,2, pp.\,259--318; see also: \textit{Uspekhi Mat. Nauk}, 2010, vol.\,65, no.\,2, pp.\,71--132.

\bibitem{KhRyab2017JMS} Kharlamov, M.\,P., Ryabov, P.\, E., Kharlamova, I.\,I., Topological Atlas of the Kovalevskaya-Yehia Gyrostat, \textit{Journal of Mathematical Sciences (United States)}, 2017,  vol.\,227, no.\,3, pp.\,241--386.
\end{thebibliography}
\end{document}